\documentclass[12pt, draftclsnofoot, onecolumn]{IEEEtran}
\usepackage{amsthm,amsmath,indentfirst,amssymb}
\usepackage{algorithm, algpseudocode}
\newtheorem{lemma}{Lemma}[section]
\newtheorem{theorem}[lemma]{Theorem}
\newtheorem{corollary}[lemma]{Corollary}
\newtheorem{definition}[lemma]{Definition}

\begin{document}

\title{Performance Guaranteed Approximation Algorithm for Minimum $k$-Connected $m$-Fold Dominating Set}
\author{\footnotesize Zhao Zhang$^1$\quad Jiao Zhou$^1$\quad Xiaohui Huang$^1$ \quad Ding-Zhu Du$^2$\\
    {\it\small $^1$ College of Mathematics Physics and Information Engineering,
    Zhejiang Normal University}\\
    {\it\small Jinhua, Zhejiang, 321004, China.}\\
    {\it\small $^2$ Department of Computer Science, University of Texas at Dallas}\\
    {\it\small Richardson, Texas, 75080, USA}}
\date{}
\maketitle

\begin{abstract}
To achieve an efficient routing in a wireless sensor network, connected dominating set (CDS) is used as virtual backbone. A fault-tolerant virtual backbone can be modeled as a $(k,m)$-CDS. For a connected graph $G=(V,E)$ and two fixed integers $k$ and $m$, a node set $C\subseteq V$ is a $(k,m)$-CDS of $G$ if every node in $V\setminus C$ has at least $m$ neighbors in $C$, and the subgraph of $G$ induced by $C$ is $k$-connected. Previous to this work, approximation algorithms with guaranteed performance ratio in a general graph were know only for $k\leq 3$. This paper makes a significant progress by presenting a $(2k-1)\alpha_0$ approximation algorithm for general $k$ and $m$ with $m\geq k$, where $\alpha_0$ is the performance ratio for the minimum CDS problem. Using currently best known ratio for $\alpha_0$, our algorithm has performance ratio $O(\ln\Delta)$, where $\Delta$ is the maximum degree of the graph.

\end{abstract}

\section{Introduction}

A wireless sensor network (WSN) consists of large quantities of sensors, which collaborate to accomplish a task. In implementing a WSN, saving energy is an important issue. If all sensors participate in the transmission frequently, it is clearly a waste of energy. Furthermore, such an active participation in transmission will incur a more serious problem known as {\em broadcast storm}, which is the result of intense interferences. Motivated by these considerations, people proposed using {\em virtual backbone} to play the main role of transmission \cite{Das,Ephremides}, which can be modeled as a connected dominating set in a graph.

Given a connected graph $G=(V,E)$, a node set $C\subseteq V$ is a {\em dominating set} (DS) of $G$ if every node in $V\setminus C$ has at least one neighbor in $C$. It is a {\em connected dominating set} (CDS) of $G$ if $G[C]$ (the subgraph of $G$ induced by $C$) is connected. Using CDS as virtual backbone, information can be shared in the whole network. On the other hand, since the number of nodes participating transmission is reduced, the energy is saved and the interference is lessened.

In many applications, fault-tolerance is desirable. The concept of $(k,m)$-CDS is used to describe a fault-tolerant virtual backbone \cite{Dai}. A node set $C$ is a {\em $k$-connected $m$-fold dominating set} ($(k,m)$-CDS for abbreviation) if every node in $V\setminus C$ has at least $m$-neighbors in $C$ and $G[C]$ is $k$-connected. Using a $(k,m)$-CDS as the virtual backbone, in the presence of at most $\min\{m-1,k-1\}$ faults, information can still be shared in the remaining network.

\subsection{Related Works}

Connected dominating set was proposed in \cite{Das,Ephremides} to serve as a virtual backbone of a wireless sensor network for the purpose of saving energy and alleviating interference. Guha and Khuller \cite{Guha} proved that a minimum CDS cannot be approximated within a factor of $(1-\varepsilon)\ln n$ for any $0<\varepsilon<1$ unless $NP\subseteq DTIME(n^{O(\log\log n)})$. In the same paper, Guha and Khuller gave two approximation algorithms achieving performance ratios of $2(H(\Delta)+1)$ and $H(\Delta)+2$, respectively, where $\Delta$ is the maximum degree of the graph and $H(\cdot)$ is the Harmonic number. Using a new method dealing with non-submodular potential functions, Ruan {\it et al.} \cite{Ruan} improved the ratio to $\ln\Delta+2$.

A homogeneous wireless sensor network can be modeled as a {\em unit disk graph}, in which every node of the graph corresponds to a sensor on the plane, and two nodes are adjacent in the graph if and only if the Euclidean distance between their corresponding sensors is at most one unit. For the minimum CDS problem in a unit disk graph, Cheng {\it et al.} \cite{Cheng} proposed a PTAS, which was generalized by Zhang {\it et al.} \cite{Zhang} to unit ball graph in higher dimensional space. These are centralized algorithms. Distributed algorithms can be found in \cite{Li1,LiYingshu,Wan,Wan1,Wu}. Comprehensive studies on the algorithmic aspect of CDS is collected in the monograph \cite{DuBookCDS}.

Dai and Wu \cite{Dai} were the first to propose the problem of constructing fault-tolerant virtual backbone. They presented three heuristics for $(k,k)$-CDS, without giving theoretical analysis on the performance ratio. After that, a lot of works appear in this field. The state of art studies are summarized in Table \ref{tab1}. Since we are now talking about fault-tolerant virtual backbone, at least one of $k$ and $m$ is assumed to be at least $2$. Except for \cite{Wang}, all other results in this table assumes $m\geq k$. It can be seen from Table \ref{tab1} that when $k\leq 3$, $(\ln\Delta+o(\ln\Delta))$-approximation exists for $(k,m)$-CDS in a general graph, which is asymptotically best possible because of the inapproximability of this problem \cite{Guha}. As to $(k,m)$-CDS on unit disk graph, constant approximation exists when $k\leq 3$. Recently, we \cite{Shi1} and Fukunage \cite{Fukunage} independently achieved constant ratio for the minimum {\em weight} $(k,m)$-CDS problem on a unit disk graph when $m\geq k$. As far as we know, there is no performance guaranteed approximation algorithm for the minimum $(k,m)$-CDS problem on a general graph.

\renewcommand{\arraystretch}{1.3}

\begin{table}[h!]
\centering
\begin{tabular}{ |c|c|c|c| }
\hline
graph & $(k,m)$ & ratio & reference \\
\hline
general & $(1,m)$ & $2H(\Delta+m-1)$ & \cite{Zhang1} \\
\hline
general & $(1,m)$ & $2+\ln(\Delta+m-2)$ & \cite{Zhou} \\
\hline
general & $(2,m)$ & $\begin{array}{c}2+\alpha_0+2\ln\alpha_0\\ \mbox{where $\alpha_0$ is performance ratio for $(1,m)$-CDS}\end{array}$ & \cite{Shi} \\
\hline
general & $(3,m)$ & $\begin{array}{c}
\left\{\begin{array}{l}\alpha_1+8+2\ln(2\alpha_1-6),\ \mbox{for}\  \alpha_1\geq4\\ 3\alpha_1+2\ln2,\ \mbox{for}\  \alpha_1<4\end{array}\right.\\ \mbox{where $\alpha_1$ is performance ratio for $(2,m)$-CDS} \end{array}$ & \cite{Zhang2}\\
\hline
general & $(k,m)$ & $\begin{array}{c}(2k-1)\alpha_0\\ \mbox{where $\alpha_0$ is performance ratio for $(1,m)$-CDS}\end{array}$ & *\\
\hline
UDG & $(2,1)$ & 72 & \cite{Wang} \\
\hline
UDG & $(1,m)$ & $\left\{\begin{array}{ll}5+5/m, & m\leq 5\\ 7, & m>5\end{array}\right.$ & \cite{Shang}\\
\hline
UDG & $(2,m)$ & $\left\{\begin{array}{ll}15+15/m, & 2\leq m\leq 5\\ 21, & m>5\end{array}\right.$ & \cite{Shang}\\
\hline
UDG & $(2,m)$ & $\left\{\begin{array}{l}7+5/m+2\ln (5+5/m),\ \mbox{for}\  2\leq m\leq 5\\ 12.89,\qquad\qquad\mbox{for}\  m>5\end{array}\right.$ & \cite{Shi}\\
\hline
UDG & $(3,m)$ & constant (280 for $m=3$) & \cite{Wang1}\\
\hline
UDG & $(3,m)$ & $\begin{array}{ll}\mbox{$5\alpha_1^U$, where $\alpha_1^U$ is performance ratio for $(2,m)$-CDS}\\ \mbox{on UDG (62.5 for $m=3$)}\end{array}$ & \cite{Wang2}\\
\hline
UDG & $(3,m)$ & $\left\{\begin{array}{ll} 26.34, & m=3\\ 25.68, & m=4\\26.86, & m\geq 5\end{array}\right.$ & \cite{Zhang2}\\
\hline
UDG & $(k,m)$ & minmum weight version: constant & \cite{Fukunage,Shi1}\\
\hline
\end{tabular}
\vskip 0.2cm \caption{Results on $(k,m)$-CDS with guaranteed performance ratio. In this table, $m\geq k$. * indicates the result obtained in this paper.}
\label{tab1}
\end{table}

\subsection{Our Contribution}

In this paper, we present an approximation algorithm for the minimum $(k,m)$-CDS problem for arbitrary constants $k$ and $m$ with $m\geq k$, which achieves performance ratio $(2k-1)\alpha_0$, where $\alpha_0$ is the performance ratio for the minimum CDS problem. Using the best known ratio $\alpha_0=2+\ln (\Delta+m-2)$ for a general graph, our algorithm has performance ratio $O(\ln\Delta)$. To the best of our knowledge, this is the first approximation algorithm with a guaranteed performance ratio for $k\geq 4$ in a general graph. For unit disk graph, using the best known ratio $\alpha_0=1+\varepsilon$, where $\varepsilon$ is an arbitrary real number in $(0,1)$, our algorithm achieves a constant performance ratio of $(2k-1+\varepsilon)$, which is smaller than the constant ratio in \cite{Fukunage,Shi1} (it should be noticed that our algorithm only works for the unweighted case while those in \cite{Fukunage,Shi1} are valid for the weighted case).

In previous works \cite{Shi,Zhang2} on approximation algorithms for $(k,m)$-CDS in a general graph with $k\leq 3$, the main idea is to augment an $(i,m)$-CDS into an $(i+1,m)$-CDS by extending a maximal $(i+1)$-connected subgraph. However, for $i\geq 3$, such an idea has not succeeded in finding a performance guaranteed approximation because for $i\geq 3$, the changes of structures are very complicated when more nodes are added. In this paper, we use a new idea of extending a so-called $i$-block (the subgraph induced by which might even be disconnected).

The outline of our algorithm is as follows. The algorithm starts from a $(1,m)$-CDS $C_0$ of $G$ and increases its connectivity iteratively. We shall show that for any $1\leq i\leq k-1$, it is possible to make an $(i,m)$-CDS to be an $(i+1,m)$-CDS in polynomial time by adding at most $2|C_0|$ nodes. Unlike previous works which makes use of {\em brick decomposition} \cite{Zhou,Shi,Zhang2}, the strategy of this paper is to expand a so-called {\em $i$-block} step by step, where an $i$-block is a maximal set of nodes which cannot be separated by any $i$-separator of the subgraph induced by current backbone nodes \cite{Mader}. Notice that the subgraph induced by an $i$-block can even be disconnected. However, we can show that by adding at most $2|C_0|$ nodes, an $i$-block can be expanded into an $(i+1)$-connected graph. This is achieved by showing that as long as current backbone nodes do not induce an $(i+1)$-connected graph, it is always possible to find at most two nodes in polynomial time, the addition of which strictly expands an $i$-block $B$ by merging at least one node from $C_0\setminus B$. Such expansion can be executed at most $|C_0|$ times, resulting in the addition of at most $2|C_0|$ nodes. Then, the desired performance ratio follows.

The paper is organized as follows. In Section \ref{Pre}, some preliminary results are given. Section \ref{algo} presents the algorithm and analyzes its performance ratio. Section \ref{con} concludes the paper and points out some directions for future work.

\section{Preliminary Results}\label{Pre}

First, we introduce some notations used in this paper. For a node $u\in V(G)$ and a node set $C\subseteq V(G)$, $N_C(v)$ is the set of nodes in $C$ which are adjacent with $u$ in $G$. In particular, $N_G(u)=N_{V(G)}(u)$ is the {\em neighbor set} of $u$ in $G$. For a node set $U\subseteq V(G)$, $N_G(U)=\left(\bigcup_{u\in U}N_G(u)\right)\setminus U$ is the {\em open neighbor set} of $u$, and $N_G[U]=N_G(U)\cup U$ is the {\em closed neighbor set} of $U$. Suppose $G$ is a connected graph. A node set $S\subseteq V(G)$ is called a {\em separator} of $G$ if $G-S$ is no longer connected. In particular, a separator of $G$ with cardinality $i$ is called an {\em $i$-separator} of $G$.

The following result is a classic result in graph theory \cite{Bondy}.

\begin{lemma}\label{lem15-3-30-2}
Suppose $G_1$ is an $i$-connected graph and
$G_2$ is obtained from $G_1$ by adding a new
vertex $u$ and joining $u$ to at least $i$ vertices
of $G_1$. Then $G_2$ is also $i$-connected.
\end{lemma}

\begin{corollary}\label{cor15-3-30-3}
Suppose $i$ and $m$ are two positive integers with $m\geq i+1$, $G$ is an $(i+1)$-connected graph, and $C$ is an $(i',m)$-CDS of $G$ with $i'\leq i$. For any node set $U\subseteq V(G)\setminus C$,

$(\romannumeral1)$ node set $C\cup U$
is also an $(i',m)$-CDS of $G$;

$(\romannumeral2)$ if $S$ is a separator of $G[C\cup U]$ with cardinality at most $i$, then $S\cap C$ is also a separator of $G[C]$;

$(\romannumeral3)$ if $G[C]$ is $i$-connected and $S$ is an $i$-separator of $G[C\cup U]$, then $S\subseteq C$ and $S$ is also an $i$-separator of $G[C]$.
\end{corollary}
\begin{proof}
Property $(\romannumeral1)$ is the result of Lemma \ref{lem15-3-30-2} and the monotonicity of $m$-fold domination.

$(\romannumeral2)$ Suppose $S\cap C$ is not a separator of $G[C]$. Then $G[C]-(S\cap C)$ is connected. Since $C$ is an $m$-fold dominating set, any node $u\in U\setminus S$ has $|N_C(u)|\geq m\geq i+1>|S|$, and thus $u$ has a neighbor in $C\setminus S$. It follows that $G[C\cup U]-S$ is connected, contradicting that $S$ is a separator of $G[C\cup U]$.

$(\romannumeral3)$ is a direct consequence of $(\romannumeral2)$.
\end{proof}

Property $(\romannumeral3)$ implies that adding nodes into an $(i,m)$-CDS will not create new $i$-separators.

For a node set $A$, if $V\setminus N_G[A]\neq\emptyset$, then $N_G(A)$ is a separator. In particular, if $N_G(A)$ is a minimum separator of $G$, then $A$ is called a {\em fragment} of $G$. The following inequality is well known (one may find it, for example, in \cite{Bondy}): for any two node sets $A$ and $B$,
\begin{equation}\label{eq10-4-1}
|N_G(A)|+|N_G(B)|\geq |N_G(A\cap B)|+|N_G(A\cup B)|.
\end{equation}

\begin{lemma}\label{lem-fragment}
Suppose $A_1,A_2$ are two fragments of $G$. If $A_1\cap A_2\neq\emptyset$ and $V\setminus (N_G[A_1]\cup N_G[A_2])\neq\emptyset$, then both $A_1\cap A_2$ and $A_1\cup A_2$ are also fragments of $G$.
\end{lemma}
\begin{proof}
By the condition of this lemma, both $N_G(A_1\cap A_2)$ and $N_G(A_1\cup A_2)$ are separators of $G$. Suppose the connectivity of $G$ is $i$, then $|N_G(A_1\cap A_2)|\geq i$ and $|N_G(A_1\cup A_2)|\geq i$. By the definition of fragment, $|N_G(A_1)|=|N_G(A_2)|=i$. Then by inequality \eqref{eq10-4-1},
\begin{equation}\label{eq10-4-2}
2i=|N_G(A_1)|+|N_G(A_2)|\geq |N_G(A_1\cap A_2)|+|N_G(A_1\cup A_2)|\geq 2i,
\end{equation}
which implies that $|N_G(A_1\cap A_2)|=|N_G(A_1\cup A_2|=i$. The lemma follows.
\end{proof}

For two node sets $S,B\subseteq V(G)$, we say that $S$ cannot not separate $B$ if $G[B\setminus S$] is connected. In the following, we give the definition of $i$-block, which was first studied by Mader \cite{Mader}, and receives a lot of attention recently because they offer a meaningful notion of the ``$k$-connected pieces'' into which the graph may be decomposed \cite{Carmesin2}. This notion is related to, but not the same as, the notion of a tangle proposed by Robertson and Seymour \cite{Robertson} in the study of minor.

\begin{definition}[$i$-block]\label{def-block}
Suppose $G$ is an $i$-connected graph. A node set $B$ is an {\em $i$-block} of $G$ if $B$ is a maximal node set satisfying the following two properties:

$(\romannumeral1)$ $|B|\geq i+1$, and

$(\romannumeral2)$ there is no $i$-separator of $G$ which can separate $B$.
\end{definition}

By Menger's Theorem, property $(\romannumeral2)$ in the above definition is equivalent to say that every pair of nodes in $B$ are connected by at least $(i+1)$ internally disjoint paths in $G$. Notice that the subgraph of $G$ induced by an $i$-block $B$ might even be disconnected. Those internally disjoint paths are required to be in $G$, not in $G[B]$. An example of a $3$-block is given in Fig.\ref{fig2}.

\begin{figure*}[!htbp]
\begin{center}
\begin{picture}(100,110)
\put(45,-10){(a)} \put(77.7,61.5){\circle*{5}}
\put(61.5,77.7){\circle*{5}} \put(38.5,77.7){\circle*{5}}
\put(22.3,61.5){\circle*{5}} \put(22.3,38.5){\circle*{5}}
\put(38.5,22.3){\circle*{5}} \put(61.5,22.3){\circle*{5}}
\put(77.7,38.5){\circle*{5}} \put(96.2,69.1){\circle*{5}}
\put(69.1,96.2){\circle*{5}} \put(30.9,96.2){\circle*{5}}
\put(3.8,69.1){\circle*{5}} \put(3.8,30.9){\circle*{5}}
\put(30.9,3.8){\circle*{5}} \put(69.1,3.8){\circle*{5}}
\put(96.2,30.9){\circle*{5}}
\qbezier(77.7,61.5)(69.6,69.6)(61.5,77.7)
\qbezier(61.5,77.7)(50.0,77.7)(38.5,77.7)
\qbezier(38.5,77.7)(30.4,69.6)(22.3,61.5)
\qbezier(22.3,61.5)(22.3,50.0)(22.3,38.5)
\qbezier(22.3,38.5)(30.4,30.4)(38.5,22.3)
\qbezier(38.5,22.3)(50.0,22.3)(61.5,22.3)
\qbezier(61.5,22.3)(69.6,30.4)(77.7,38.5)
\qbezier(77.7,38.5)(77.7,50.0)(77.7,61.5)
\qbezier(96.2,69.1)(82.7,82.7)(69.1,96.2)
\qbezier(69.1,96.2)(50.0,96.2)(30.9,96.2)
\qbezier(30.9,96.2)(17.4,82.7)(3.8,69.1)
\qbezier(3.8,69.1)(3.8,50.0)(3.8,30.9)
\qbezier(3.8,30.9)(17.4,17.4)(30.9,3.8)
\qbezier(30.9,3.8)(50.0,3.8)(69.1,3.8)
\qbezier(69.1,3.8)(82.7,17.4)(96.2,30.9)
\qbezier(96.2,30.9)(96.2,50.0)(96.2,69.1)
\qbezier(77.7,61.5)(87.0,65.3)(96.2,69.1)
\qbezier(61.5,77.7)(65.3,87.0)(69.1,96.2)
\qbezier(38.5,77.7)(34.7,87.0)(30.9,96.2)
\qbezier(22.3,61.5)(13.1,65.3)(3.8,69.1)
\qbezier(22.3,38.5)(13.1,34.7)(3.8,30.9)
\qbezier(38.5,22.3)(34.7,13.1)(30.9,3.8)
\qbezier(61.5,22.3)(65.3,13.1)(69.1,3.8)
\qbezier(77.7,38.5)(87.0,34.7)(96.2,30.9)
\qbezier(61.5,77.7)(46.2,87.0)(30.9,96.2)
\qbezier(38.5,77.7)(53.8,87.0)(69.1,96.2)
\qbezier(61.5,22.3)(46.2,13.1)(30.9,3.8)
\put(27,100){$u_1$}\put(65,100){$u_2$}\put(55,68){$u_3$}\put(38,68){$u_4$}\put(54,27){$u_5$}\put(25,-5){$u_6$}
\end{picture}

\vskip 0.5cm\caption{Node set $U=\{u_1,u_2,u_3,u_4,u_5,u_6\}$ forms a $3$-block. Notice that $G[U]$ is even disconnected. Also notice that $U'=\{u_1,u_2,u_3,u_4\}$ is not a $3$-block because it is not maximal with respect to the second property in Definition \ref{def-block}.}\label{fig2}
\end{center}
\end{figure*}
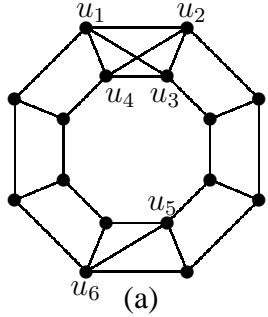

We say that $U$ is a {\em connected nodes set of $G$} if $G[U]$ is connected.

\begin{lemma}\label{lem15-10-14-1}
Suppose $i,m$ are two positive integers with $m\geq i+1$, $C$ is an $(i,m)$-CDS of $G$, $U\subseteq V(G)\setminus C$ is a connected node set of $G$. Then no $i$-separator of $G[C\cup U]$ can separate node set $U\cup N_C(U)$.
\end{lemma}
\begin{proof}
Consider an arbitrary $i$-separator $S$ of $G[C\cup U]$. By Corollary \ref{cor15-3-30-3} $(\romannumeral3)$, $S\subseteq C$. In other words, $S\cap U=\emptyset$. Combining this with the assumption that $G[U]$ is connected, we see that subgraph $G[U\cup N_C(U)]-S$ is connected.
\end{proof}

\section{Algorithm and Its Analysis}\label{algo}

\subsection{The Highest Level Outline of Our Algorithm}

Our algoriothm starts from a $(1,m)$-CDS of $G$, and then augments its connectivity iteratively by adding more nodes. Suppose the $(1,m)$-CDS is $C_0$, which can be found using an existing algorithm such as the one in \cite{Zhou}. For $i=1,2,\ldots,k-1$, in the $i$-th iteration, suppose an $(i,m)$-CDS is already at hand, say $C$, we shall show that it is possible to find at most $2|C_0|$ nodes in polynomial time, such that adding them into $C$ will result in an $(i+1,m)$-CDS. As a consequence, a $(k,m)$-CDS can be obtained in polynomial time whose size is at most $(2k-1)|C_0|$. Then, the desired performance ratio follows from the upper bound of $|C_0|$ in \cite{Zhou}.

\subsection{Augmenting an $(i,m)$-CDS into an $(i+1,m)$-CDS}

In the following, we focus on how to augment an $(i,m)$-CDS into an $(i+1,m)$-CDS, where $i<k$. The idea is to expand an $i$-block iteratively until an $(i+1,m)$-CDS is obtained. To be more concrete, suppose $C$ is an $(i,m)$-CDS, and $B$ is an $i$-block of $G[C]$. We shall show that as long as $G[C]$ is not $(i+1)$-connected, it is possible to find a node set $U$ in polynomial time such that $|U|\leq 2$, $B$ is contained in a {\em strictly bigger} $i$-block of $G[C\cup U]$, say $B'$, and $B'\setminus B$ contains at least one node of $C_0$. If $G[C\cup U]$ is still not $(i+1)$-connected, then we set $C\leftarrow C\cup U$, $B\leftarrow B'$, and repeat. Because in each iteration, at least one node of $C_0$ is newly merged into the expanded $i$-block, such a process can be executed at most $|C_0|$ times. When it terminates, we must have an $(i+1)$-connected subgraph, and the number of nodes added is at most $2|C_0|$.

To realize the above idea, one question is: what if $G[C]$ does not contain an $i$-block. We shall deal with this problem later, and first deal with the case when $G[C]$ already contains an $i$-block.

\subsubsection{Expanding an $i$-Block}\label{expand-block}

Suppose $C$ is an $(i,m)$-CDS and $G[C]$ is not $(i+1)$-connected. Let $S_0$ be an $i$-separator of $G[C]$. Since $B$ is an $i$-block which cannot be separated by any $i$-separator, $B\setminus S_0$ is contained in one connected component of $G[C]-S_0$, denote this connected component as $G_1^{S_0}$, and let $G_2^{S_0}$ be the union of the other connected components of $G[C]-S_0$ (see Fig.\ref{figA1}$(a)$). By Corollary \ref{cor15-3-30-3} $(\romannumeral2)$, $S_0\cap C_0$ is a separator of $G[C_0]$. Let $G_1$ be the union of those connected components of $G[C_0]-(S_0\cap C_0)$ which have nonempty intersections with $G_1^{S_0}$, and let $G_2$ be the union of those remaining connected components of $G[C_0]-(S_0\cap C_0)$. We claim that neither $V(G_1)$ nor $V(G_2)$ is empty. In fact, if $V(G_2)$ is empty, then $C_0$ is contained in $V(G_1^{S_0})\cup S_0$. Consider a node $u\in V(G_2^{S_0})$, it is adjacent with at least $m\geq k>i=|S_0|$ nodes of $C_0$, and thus has at least one neighbor in $V(G_1^{S_0})$, contradicting that $G_1^{S_0}$ and $G_2^{S_0}$ are disconnected by the removal of $S_0$. The case when $V(G_1)$ is empty can be argued similarly.

\begin{figure*}[!htbp]
\begin{center}
\hskip -0.5cm\begin{picture}(130,130)
\qbezier(92.0,100.0)(91.5,103.9)(89.9,107.7)
\qbezier(89.9,107.7)(87.3,111.1)(83.8,114.1)
\qbezier(83.8,114.1)(79.6,116.6)(74.7,118.5)
\qbezier(74.7,118.5)(69.5,119.6)(64.0,120.0)
\qbezier(64.0,120.0)(58.5,119.6)(53.3,118.5)
\qbezier(53.3,118.5)(48.4,116.6)(44.2,114.1)
\qbezier(44.2,114.1)(40.7,111.1)(38.1,107.7)
\qbezier(38.1,107.7)(36.5,103.9)(36.0,100.0)
\qbezier(36.0,100.0)(36.5,96.1)(38.1,92.3)
\qbezier(38.1,92.3)(40.7,88.9)(44.2,85.9)
\qbezier(44.2,85.9)(48.4,83.4)(53.3,81.5)
\qbezier(53.3,81.5)(58.5,80.4)(64.0,80.0)
\qbezier(64.0,80.0)(69.5,80.4)(74.7,81.5)
\qbezier(74.7,81.5)(79.6,83.4)(83.8,85.9)
\qbezier(83.8,85.9)(87.3,88.9)(89.9,92.3)
\qbezier(89.9,92.3)(91.5,96.1)(92.0,100.0)

\qbezier(54.0,98.2)(50.4,99.7)(46.3,100.2)
\qbezier(46.3,100.2)(42.0,99.7)(37.7,98.2)
\qbezier(37.7,98.2)(33.4,95.8)(29.4,92.5)
\qbezier(29.4,92.5)(25.7,88.5)(22.7,84.0)
\qbezier(22.7,84.0)(20.3,79.1)(18.6,74.0)
\qbezier(18.6,74.0)(17.8,68.8)(17.9,63.9)
\qbezier(17.9,63.9)(18.7,59.4)(20.4,55.4)
\qbezier(20.4,55.4)(22.9,52.2)(26.0,49.8)
\qbezier(26.0,49.8)(29.6,48.3)(33.7,47.8)
\qbezier(33.7,47.8)(38.0,48.3)(42.3,49.8)
\qbezier(42.3,49.8)(46.6,52.2)(50.6,55.5)
\qbezier(50.6,55.5)(54.3,59.5)(57.3,64.0)
\qbezier(57.3,64.0)(59.7,68.9)(61.4,74.0)
\qbezier(61.4,74.0)(62.2,79.2)(62.1,84.1)
\qbezier(62.1,84.1)(61.3,88.6)(59.6,92.6)
\qbezier(59.6,92.6)(57.1,95.8)(54.0,98.2)

\qbezier(47.8,67.8)(43.6,71.3)(38.7,73.9)
\qbezier(38.7,73.9)(33.5,75.5)(28.0,76.0)
\qbezier(28.0,76.0)(22.5,75.5)(17.3,73.9)
\qbezier(17.3,73.9)(12.4,71.3)(8.2,67.8)
\qbezier(8.2,67.8)(4.7,63.6)(2.1,58.7)
\qbezier(2.1,58.7)(0.5,53.5)(0.0,48.0)
\qbezier(0.0,48.0)(0.5,42.5)(2.1,37.3)
\qbezier(2.1,37.3)(4.7,32.4)(8.2,28.2)
\qbezier(8.2,28.2)(12.4,24.7)(17.3,22.1)
\qbezier(17.3,22.1)(22.5,20.5)(28.0,20.0)
\qbezier(28.0,20.0)(33.5,20.5)(38.7,22.1)
\qbezier(38.7,22.1)(43.6,24.7)(47.8,28.2)
\qbezier(47.8,28.2)(51.3,32.4)(53.9,37.3)
\qbezier(53.9,37.3)(55.5,42.5)(56.0,48.0)
\qbezier(56.0,48.0)(55.5,53.5)(53.9,58.7)
\qbezier(53.9,58.7)(51.3,63.6)(47.8,67.8)

\qbezier(119.8,67.8)(115.6,71.3)(110.7,73.9)
\qbezier(110.7,73.9)(105.5,75.5)(100.0,76.0)
\qbezier(100.0,76.0)(94.5,75.5)(89.3,73.9)
\qbezier(89.3,73.9)(84.4,71.3)(80.2,67.8)
\qbezier(80.2,67.8)(76.7,63.6)(74.1,58.7)
\qbezier(74.1,58.7)(72.5,53.5)(72.0,48.0)
\qbezier(72.0,48.0)(72.5,42.5)(74.1,37.3)
\qbezier(74.1,37.3)(76.7,32.4)(80.2,28.2)
\qbezier(80.2,28.2)(84.4,24.7)(89.3,22.1)
\qbezier(89.3,22.1)(94.5,20.5)(100.0,20.0)
\qbezier(100.0,20.0)(105.5,20.5)(110.7,22.1)
\qbezier(110.7,22.1)(115.6,24.7)(119.8,28.2)
\qbezier(119.8,28.2)(123.3,32.4)(125.9,37.3)
\qbezier(125.9,37.3)(127.5,42.5)(128.0,48.0)
\qbezier(128.0,48.0)(127.5,53.5)(125.9,58.7)
\qbezier(125.9,58.7)(123.3,63.6)(119.8,67.8)

\put(44,40){\circle*{3}}\put(51,37){\circle*{3}}\put(64,35){\circle*{3}}\put(84,40){\circle*{3}}
\qbezier(44,40)(64,30)(84,40)
\qbezier(64,105)(30,105)(15,45)\qbezier(15,45)(18,38)(25,47)\qbezier(25,47)(60,120)(40,45)\qbezier(40,45)(42,26)(50,45)
\qbezier(50,45)(64,140)(78,45)
\qbezier(64,105)(98,105)(113,45)\qbezier(113,45)(110,38)(103,47)\qbezier(103,47)(68,120)(88,45)\qbezier(88,45)(86,26)(78,45)
\put(27,110){$S_0$}\put(15,85){$B$}\put(98,85){$C_0$}
\put(27,10){$G_1^{S_0}$}\put(100,10){$G_2^{S_0}$}
\put(62,20){$P$}\put(30,37){$u_0$}\put(45,30){$u_1$}\put(60,38){$u_2$}\put(80,30){$u_3$}
\put(60,-5){(a)}
\end{picture}
\hskip 0.5cm\begin{picture}(100,130)
\qbezier(65.0,70.0)(64.1,76.5)(61.7,82.5)
\qbezier(61.7,82.5)(57.7,87.7)(52.5,91.7)
\qbezier(52.5,91.7)(46.5,94.1)(40.0,95.0)
\qbezier(40.0,95.0)(33.5,94.1)(27.5,91.7)
\qbezier(27.5,91.7)(22.3,87.7)(18.3,82.5)
\qbezier(18.3,82.5)(15.9,76.5)(15.0,70.0)
\qbezier(15.0,70.0)(15.9,63.5)(18.3,57.5)
\qbezier(18.3,57.5)(22.3,52.3)(27.5,48.3)
\qbezier(27.5,48.3)(33.5,45.9)(40.0,45.0)
\qbezier(40.0,45.0)(46.5,45.9)(52.5,48.3)
\qbezier(52.5,48.3)(57.7,52.3)(61.7,57.5)
\qbezier(61.7,57.5)(64.1,63.5)(65.0,70.0)
\put(100,10){\line(-1,0){90}}\put(100,40){\line(-1,0){90}}\put(100,70){\line(-1,0){90}}\put(100,100){\line(-1,0){90}}
\put(10,10){\line(0,1){90}}\put(40,10){\line(0,1){90}}\put(70,10){\line(0,1){90}}\put(100,10){\line(0,1){90}}
\put(85,25){\circle*{3}}
\put(-5,80){$G_1^S$}\put(0,50){$S$}\put(-5,20){$G_2^S$}
\put(20,105){$G_1^{S_0}$}\put(50,105){$S_0$}\put(80,105){$G_2^{S_0}$}
\put(30,72){$B$}\put(85,17){$u_t$}
\put(50,-5){(b)}
\end{picture}\hskip 0.5cm\begin{picture}(100,130)
\qbezier(65.0,70.0)(64.1,76.5)(61.7,82.5)
\qbezier(61.7,82.5)(57.7,87.7)(52.5,91.7)
\qbezier(52.5,91.7)(46.5,94.1)(40.0,95.0)
\qbezier(15.0,70.0)(15.9,63.5)(18.3,57.5)
\qbezier(18.3,57.5)(22.3,52.3)(27.5,48.3)
\qbezier(27.5,48.3)(33.5,45.9)(40.0,45.0)
\qbezier(40.0,45.0)(46.5,45.9)(52.5,48.3)
\qbezier(52.5,48.3)(57.7,52.3)(61.7,57.5)
\qbezier(61.7,57.5)(64.1,63.5)(65.0,70.0)
\put(100,10){\line(-1,0){90}}\put(100,40){\line(-1,0){90}}\put(100,70){\line(-1,0){90}}\put(100,100){\line(-1,0){60}}
\put(10,10){\line(0,1){60}}\put(40,10){\line(0,1){90}}\put(70,10){\line(0,1){90}}\put(100,10){\line(0,1){90}}
\put(85,25){\circle*{3}}\put(26,64){\circle*{3}}\put(35,56){\circle*{3}}\put(50,60){\circle*{3}}
\put(80,55){\circle*{3}}\put(55,25){\circle*{3}}\put(55,15){\circle*{3}}
\put(-5,80){$G_1^S$}\put(0,50){$S$}\put(-5,20){$G_2^S$}
\put(20,105){$G_1^{S_0}$}\put(50,105){$S_0$}\put(80,105){$G_2^{S_0}$}
\put(25,75){$B$}\put(85,17){$u_t$}
\put(16,63){$v_1$}\put(24,52){$v_2$}\put(57,17){$v'_1$}\put(50,30){$v'_2$}
\qbezier(85,25)(80,30)(85,35)\qbezier(80,55)(90,45)(85,35)\qbezier(80,55)(55,65)(55,80)
\qbezier(85,25)(70,30)(50,60)\qbezier(50,80)(45,70)(50,60)
\qbezier(85,25)(70,30)(55,25)\qbezier(35,40)(40,20)(55,25)\qbezier(35,40)(30,50)(45,80)
\qbezier(85,25)(70,10)(55,15)\qbezier(30,35)(30,25)(55,15)\qbezier(30,35)(10,60)(45,85)
\put(50,-5){(c)}
\end{picture}

\vskip 0.5cm\caption{$(a)$ An illustration of how to find node set $U$ to be added. The starfish-shaped region represents $C_0$. $(b)$ The distribution of $B$ and $u_t$. $(c)$ An illustration for the proof of property $(P_3)$.}\label{figA1}
\end{center}
\end{figure*}
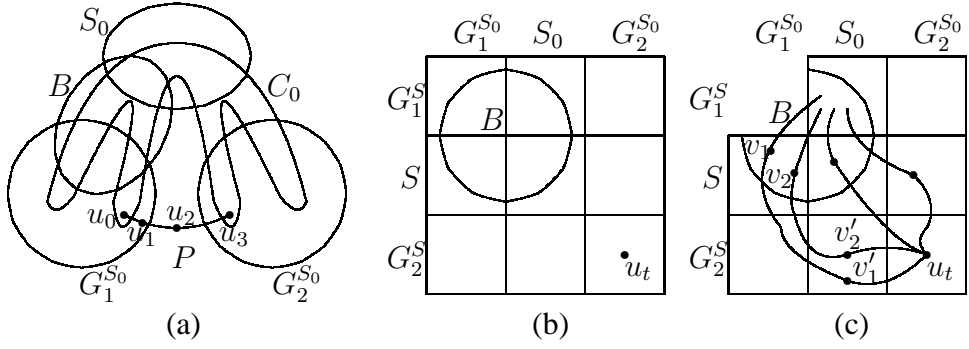

Since $G$ is $k$-connected and $|S_0|=i<k$, graph $G-S_0$ is connected. Let $P$ be a shortest path $P$ between $G_1$ and $G_2$ in $G-S_0$. Suppose $P=u_0u_1\ldots u_t$, where $u_0\in V(G_1)$ and $u_t\in V(G_2)$. We claim that $t\leq 3$. In fact, if $t\geq 4$, consider node $u_2$. Since $C_0$ is an $m$-fold dominating set with $m\geq k>|S_0|$, node $u_2$ is adjacent with some node $u\in C_0\setminus S_0$. If $u\in V(G_1)$, then $uu_2\ldots u_t$ is a shorter path between $G_1$ and $G_2$. If $u\in V(G_2)$, then $u_0u_1u_2u$ is a shorter path between $G_1$ and $G_2$. Both contradict the choice of $P$. As a consequence, if we set $U$ to be the set of internal nodes of $P$, then $|U|\leq 2$.

Notice that $u_t\in C_0\setminus B$. So, if
\begin{equation}\label{eq15-10-14-5}
\mbox{no $i$-separator of $G[C\cup U]$ can separate $B\cup\{u_t\}$,}
\end{equation}
then we have found a desired $U$ and may continue with the next iteration. Otherwise, there exists an $i$-separator $S$ of $G[C\cup U]$ which separates $u_t$ from a node of $B$. Since $B\setminus S$ is connected in $G[C]-S$,
\begin{equation}\label{eq11-13-1}
\mbox{$S$ separates $u_t$ from all nodes of $B\setminus S$.}
\end{equation}
In the following, we shall consider nodes in $B\setminus S$ as a whole.

By Corollary \ref{cor15-3-30-3} $(\romannumeral3)$, $S$ is an $i$-separator of $G[C]$. Let $G_1^S$ be a connected component of $G[C]-S$ which contains $B\setminus S$ (recall that $B$ cannot be separated by $S$). We have the following properties.

$(P_1)$ The distribution of $B$ is as in Fig.\ref{figA1}$(b)$, since $B\cap V(G_2^{S_0})=\emptyset$ and $B\cap V(G_2^S)=\emptyset$.

$(P_2)$ Node $u_t\in V(G_2^{S_0})\cap V(G_2^S)$. The reason why $u_t\in V(G_2^{S_0})$ is by the choice of path $P$. The reason why $u_t\in V(G_2^S)$ is because of \eqref{eq11-13-1}.

$(P_3)$ $V(G_1^{S_0})\cap V(G_1^S)\neq\emptyset$. Suppose this is not true, then the distribution of $B$ is as in Fig.\ref{figA1}$(c)$. It follows that $B\cap V(G_1^{S_0})\subseteq S$. Since $|B|\geq i+1>|S_0|$, we have
\begin{equation}\label{eq15-10-14-2}
|B\cap V(G_1^{S_0})|=|B\setminus S_0|>|S_0\setminus B|.
\end{equation}
Since $S$ is an $i$-separator of $G[C]$ separating $u_t$ from $B\setminus S$, there are $i$ internally disjoint paths in $G[C]$ connecting $B\setminus S$ and $u_t$, each path containing exactly one node of $S$. For a node $v\in S$, denote by $P_v$ such a path containing $v$. For each node
$v\in B\cap V(G_1^{S_0})\subseteq S$, the path $P_v$ must go through a distinct node $v'\in S_0\setminus B$ (see Fig.\ref{figA1}$(c)$). But this is impossible because of \eqref{eq15-10-14-2}. Property $(P_3)$ is proved.

Taking $A_1=V(G_1^{S_0})$, $A_2=V(G_1^S)$. Property $(P_3)$ shows that $A_1\cap A_2\neq\emptyset$. Property $(P_2)$ shows that $u_t\in V(G_2^{S_0})\cap V(G_2^S)=C\setminus (N_G[A_1]\cup N_G[A_2])\neq\emptyset$. So by Lemma \ref{lem-fragment}, the neighbor set of $V(G_1^{S_0})\cap V(G_1^S)$, denoted as $S_1$, is also an $i$-separator of $G[C]$. Let $G_1^{S_1}$ be the connected component of $G[C]-S_1$ which contains $B\setminus S_1$ (the left-top square of Fig.\ref{figA1}$(b)$ is $G_1^{S_1}$). Notice that $S\cap V(G_1^{S_0})\neq\emptyset$, otherwise $B$ cannot be separated from $u_0$ by $S$, and thus cannot be separated from $u_t$ by $S$ (since $S\cap U=\emptyset$, $u_0$ is connected to $u_t$ through $U$ in $G[C\cup U]-S$). It follows that $|V(G_1^{S_1})|<|V(G_1^{S_0})|$.

To sum up, if for $S_0$, the shortest path $P$ found in the above way does not satisfy property \eqref{eq15-10-14-5}, then another $k$-separator $S_1$ of $G[C]$ can be found such that $G_1^{S_1}$ is strictly smaller than $G_1^{S_0}$. Replace $S_0$ by $S_1$ and repeat, such a procedure can be executed at most $|V(G_1^{S_0})|<n$ times. When it terminates, we are in the situation that property \eqref{eq15-10-14-5} is satisfied, at which time a desired node set $U$ is found, and we may continue with the next iteration.

\subsubsection{When There Is No $i$-Block}\label{no-block}

Now, we consider the situation that there is no $i$-block in $G[C]$. Suppose $C$ is an $(i,m)$-CDS and $G[C]$ is not $(i+1)$-connected. Let $S_0$ be an $i$-separator of $C$. By Corollary \ref{cor15-3-30-3} $(\romannumeral2)$, $S_0\cap C_0$ is a separator of $G[C_0]$. Let $G_1$ be a connected component of $G[C_0]-(S_0\cap C_0)$ and let $G_2$ be the union of the other connected components of $G[C_0]-(S_0\cap C_0)$. Similarly to the above, a shortest path $P$ between $G_1$ and $G_2$ in $G-S_0$ has at most two internal nodes. Let $U$ be the set of internal nodes of $P$. By Lemma \ref{lem15-10-14-1}, $U\cup N_C(U)$ cannot be separated by any $i$-separator of $G[C\cup U]$. Furthermore, since $C$ is an $m$-fold dominating set of $G$, we have $|N_C(U)|\geq m\geq k$. So, $|U\cup N_C(U)|\geq k+1>i+1$. In other words, $U\cup N_C(U)$ is contained in an $i$-block of $G[C\cup U]$. This $i$-block can serve as the starting point of the subsequent iterations.

\subsubsection{To Sum Up}

The algorithm is summarized in Algorithm \ref{algo1}. The variable $flag$ is used to indicate whether we should adjust the $i$-separator $S_0$ used in the algorithm such that adding the internal nodes of the path can strictly expand current $i$-block $B$ by at least one node of $C_0$. As we have shown in the above, the inner while loop is executed $O(n)$ times, and the outer while loop is executed at most $|C_0|=O(n)$ times. Since determining whether two nodes can be separated by an $i$-separator can be done in polynomial time using a maximum flow algorithm, an $i$-block can be found in polynomial time (if it exists). In fact, this can be accomplished more efficiently using the method in \cite{Carmesin}, in which a stronger result was obtained by Carmesin {\it et al.} showing that there is an $O(\min\{i,\sqrt{n}\}\cdot m\cdot n^2)$-time algorithm which can find all $i$-blocks of a graph. All other operations can clearly be done in polynomial time. So, our algorithm has polynomial running time. The performance ratio of the algorithm is already implied in the above analysis.

\begin{algorithm}[h!]
\caption{\textbf{}}
Input: Two positive integers $k,m$ with $m\geq k$ and a $k$-connected graph $G$.

Output: A $(k,m)$-CDS $C$.
\begin{algorithmic}[1]
    \State Find a $(1,m)$-CDS $C_0$.
    \State $C\leftarrow C_0$.
    \For{$i=1$ to $k-1$}
        \If{$G[C]$ is not $(i+1)$-connected and there is no $i$-block in $G[C]$}
            \State Use the method in Section \ref{no-block} to find a node set $U$.
            \State $C\leftarrow C\cup U$.
            \State $B\leftarrow$ the $i$-block of $G[C\cup U]$ containing $U\cup N_C(U)$.
        \EndIf
        \While{$G[C]$ is not $(i+1)$-connected}
            \State $flag\leftarrow 1$.
            \State Let $S_0$ be an $i$-separator of $G[C]$.
            \While{$flag=1$}
                \State Use the method in Section \ref{expand-block} to find a path $P=u_0\ldots u_t$ in $G-S_0$.
                \State Let $U'$ be the set of internal nodes of $P$.
                \If{$B\cup\{u_t\}$ cannot be separated by any $i$-separator of $G[C\cup U']$}
                    \State $flag\leftarrow 0$.
                    \State $U\leftarrow U'$, $C\leftarrow C\cup U$.
                    \State $B\leftarrow$ the $i$-block of $G[C]$ containing $B\cup\{u_t\}$.
                \Else
                    \State Let $S$ be an $i$-separator of $G[C\cup U']$ separating $B\cup\{u_t\}$.
                    \State $S_0\leftarrow$ the neighbor set of $V(G_1^{S_0})\cap V(G_1^S)$ in $G[C]$, where $G_1^{s_0}$ and $G_1^S$ are the connected components defined in Section \ref{expand-block}.
                \EndIf
            \EndWhile
        \EndWhile
    \EndFor
\end{algorithmic}\label{algo1}
\end{algorithm}

\begin{theorem}
Suppose $k,m$ are two positive integers with $m\geq k$ and $G$ is a $k$-connected graph. A $(k,m)$-CDS of $G$ can be found in polynomial time whose size is at most $(2k-1)\alpha_0 \cdot opt$, where $opt$ is the size of an optimal solution and $\alpha_0$ is the performance ratio for $(1,m)$-CDS of $G$. Taking $\alpha_0=2+\ln (\Delta+m-2)$ which was obtained in \cite{Zhou}, the performance ratio of our algorithm is $(2k-1)(2+\ln (\Delta+m-2))$, which is $O(\ln \Delta)$ for fixed $k$ and $m$. Based on a PTAS to find a $(1,m)$-CDS in a unit disk graph, our algorithm has performance ratio $2k-1+\varepsilon$ for $(k,m)$-CDS in unit disk graph.
\end{theorem}

\section{Conclusion}\label{con}

In this paper, we proposed the first approximation algorithm with a guaranteed performance ratio for the minimum $(k,m)$-CDS problem in a general graph for general $k$ under the assumption that $m\geq k$. The performance ratio is $(2k-1)(\ln (\Delta+m-2)+2)$. Prior to this work, we have obtained $(\ln \Delta+o(\ln\Delta))$-approximation algorithms for $k\leq 3$ \cite{Shi,Zhang2,Zhou}, the strategy of which is to expand an $i$-brick greedily, where an $i$-brick is a maximal $(i+1)$-connected {\em induced} subgraph. The difficulty of generalizing such a strategy to deal with higher value of $k$ is that for $i\geq 3$, even when one could obtain a similar decomposition result using $i$-bricks, it is not clear how the decomposition structure changes when more nodes are added. In this paper, we propose using another strategy by expanding an $i$-block instead of expanding an $i$-brick. This new strategy works for any constant $k$. However, its performance ratio is not as delicate as those for $k\leq 3$. An improvement on the coefficient before $\ln \Delta$ remains to be further explored.

\section*{Acknowledgment}
This research is supported by NSFC (61222201,11531011), and Xingjiang Talent Youth Project (2013711011).

\end{document}